\documentclass[journal,twoside,web]{ieeecolor}
\usepackage{cite}
\usepackage{lcsys}
\usepackage[compact]{titlesec}
\usepackage{amsmath,amssymb,amsfonts}
\usepackage{hyperref}
\usepackage{graphicx}
\usepackage{textcomp}
\usepackage{times}
\usepackage{graphicx}
\usepackage{mathrsfs}
\usepackage{amstext} 
\usepackage{algorithm}
\usepackage{algorithmic}
\usepackage{setspace} 
\usepackage{amssymb} 
\usepackage{mathtools} 
\usepackage{verbatim}
\newcommand{\R}{\mathbb{R}}

\newtheorem{theorem}{Theorem}[section]

\newtheorem{lemma}[theorem]{Lemma}

\newtheorem{remark}[theorem]{Remark}

\newtheorem{proposition}[theorem]{Proposition}

\newcommand{\diag}{\operatorname{diag}}
\newcommand{\inte}{\operatorname{int}}
\newcommand{\trace}{\operatorname{trace}}
\newcommand{\longthmtitle}[1]{\mbox{}{\textit{(#1):}}}
\newcommand\oprocendsymbol{\hbox{$\bullet$}}
\newcommand\oprocend{\relax\ifmmode\else\unskip\hfill\fi\oprocendsymbol}
\def\eqoprocend{\tag*{$\bullet$}}

\renewcommand{\baselinestretch}{.995}
\graphicspath{{fig/}}

\usepackage{multirow}  

\newcommand{\revision}[1]{\textcolor{blue}{#1}}
\renewcommand{\revision}[1]{{#1}}

\allowdisplaybreaks

\def\BibTeX{{\rm B\kern-.05em{\sc i\kern-.025em b}\kern-.08em
    T\kern-.1667em\lower.7ex\hbox{E}\kern-.125emX}}

\begin{document}
\title{\LARGE Controller Design for Bilinear Neural Feedback Loops}
\author{Dhruv Shah \quad Jorge Cortés \thanks{D. Shah and J. Cort\'es
    are with Department of Mechanical and Aerospace Engineering, UC
    San Diego, USA, {\tt\small \{dhshah,cortes\}@ucsd.edu}}}

\maketitle

\begin{abstract}%
  This paper considers a class of bilinear systems with a neural
  network in the loop. These arise naturally when employing machine
  learning techniques to approximate general, non-affine in the input,
  control systems.  We propose a controller design framework that
  combines linear fractional representations and tools from linear
  parameter varying control to guarantee local exponential stability
  of a desired equilibrium. The controller is obtained from the
  solution of linear matrix inequalities, which can be solved offline,
  making the approach suitable for online applications. The proposed
  methodology offers tools for stability and robustness analysis of
  deep neural networks interconnected with dynamical systems.
\end{abstract}


\section{Introduction}
The application of machine learning techniques in control is becoming
popular, particularly due to the success of deep neural networks
(DNNs) as universal function approximators. DNNs can be used for,
e.g., modeling nonlinearities in the system, learning appropriate
lifting functions for a Koopman representation, modeling disturbances
and even learning controllers. Additionally, DNNs are the main
workhorse behind the modern perception pipeline used in robotics and
autonomous vehicles. Thus, dynamical systems with a neural network in
the loop, or neural feedback loops (NFLs), are becoming increasingly
prevalent. The main drawback of working with NFLs is the lack of
formal guarantees on closed-loop stability, robustness and safety,
which are crucial for deployment in practice. This paper aims to
bridge this gap by considering a class of bilinear systems with a
neural network in the loop and designing a controller with these guarantees.


\emph{Literature Review:} When considering neural feedback loops, a
majority of work deals with linear systems.  \revision{As an
  example,}~\cite{ME-GH-CS-JPH:21} proposes a computational approach
for backward reachability analysis of systems with neural networks in
feedback with LTI systems. \revision{The work}~\cite{HH-MF-MM-GJP:20}
proposes a semidefinite program (SDP) for reachability analysis of
linear time-varying NFLs
\revision{and}~\cite{CH-JF-WL-XC-QZ:19,SD-XC-SS:19} use polynomial
approximations of DNNs.  For perception-based control systems which
can be rewritten as LTI-NFLs, \cite{SD-NM-BR-VY:20} proposes an
approach to obtain robustness guarantees on the perception pipeline.
The restriction to LTI systems limits the applicability to real
systems.  Instead, \cite{
  NJ-MA-PS:24margins,NJ-MA-PS:24synthesis} consider a plant model with
a linear fractional representation (LFR) and propose a neural network
controller design with stability and dissipativity guarantees by
appropriately adding constraints to a generic reinforcement learning
algorithm. This, however, could make the learning process slow and
potentially limit the expressibility of the DNNs.  Here, we deal with
a class of generalized bilinear systems, as a first step towards more
general nonlinear systems.
The systems considered here arise, for instance, from
applying system identification techniques enhanced by deep learning,
e.g.,~\cite{GP-AA-DG-LL-AHR-THS:25} or as a result of employing
Koopman-based representations of general control systems,
cf.~\cite{MH-JC:24-auto}.
We build upon the idea
from\revision{\cite{RS-JB-FA:23,NJ-MA-PS:24synthesis}} of using
robust control techniques for dealing with nonlinearities to create a
controller design framework for non-affine systems.

\emph{Statement of Contributions:} We consider
%
%
a bilinear system with a neural network in the loop and propose a
controller design methodology for guaranteeing local exponential
stability of a desired equilibrium point. \revision{The proposed
  methodology is broadly applicable to a wide range of dynamical
  systems where neural networks are used to model the dynamics.}  The
contributions are the following.
%
%
Our first contribution is a reformulation of the system to a linear
fractional representation by appropriately recasting the
nonlinearities as uncertainties.  We characterize the uncertainty sets
using quadratic constraints. This result enables us to tackle the
stability and robustness properties of NFLs through the lens of robust
control. Our next contribution presents the controller synthesis using
linear matrix inequalities \revision{and guarantees a well-posed
  controller when the LMI is feasible}. The controller obtained after
synthesis is richer than a vanilla linear controller with certified
performance while having a convex formulation. We perform simulations
on \revision{a 4-dimensional system} with a nontrivial nonlinearity
approximated using a ReLU MLP.

\section{Preliminaries}
We review\footnote{
  We denote by $\R_{>0}$ (resp.  $\R_{\ge 0}$) the set of positive
  (resp., nonnegative) real numbers. We denote by $\{e_1,\dots,e_l\}$
  the standard basis of $\R^l$. The matrices
  $\{E_i\}_{i=1}^l \in \R^{l \times l}$ have 1 at the $(i,i)$ entry
  and zero elsewhere. Thus, $E_i z = z_i e_i$, for any $z \in \R^l$.
%
%
  We let $I_p$ denote the $p \times p$ identity matrix and
  $0_{p \times q}$ the $p \times q$ zero matrix (we drop the
  subscripts when clear from the context).  The Kronecker product of
  $A$ and $B$ is $A \otimes B$. We abbreviate $B^TAB$ with
  $\left[\star\right]^TAB $, so that $\star$ means that matrix blocks
  are inferred symmetrically.  We let $\diag [A_1,\dots,A_n]$
  represent the block-diagonal matrix with blocks $A_1,\dots,A_n$. If
  all blocks are the same, $A_i=A$, we use
  $\diag_n(A) = \diag [A,\dots,A]$. \revision{We also use the vertical
    vectorization notation
    $\text{vec}[A_1,A_2,\dots,A_n] = [A_1^T \, A_2^T \, \dots \,
    A_n^T]^T $.} We denote by $\mathbb{S}^n$ the set of $n\times n$
  symmetric matrices. For $M \in\mathbb{S}^n$, $M \succ 0$ (resp.
  $M \succeq 0 $) means $M$ is positive definite
  (resp. semidefinite). Negative (semi)definiteness is denoted
  analogously. For $\alpha,\beta \in \R$, a piecewise smooth function
  $\xi:\R \to \R$ is slope-restricted on $[\alpha,\beta]$ if
  $\alpha \le \xi'(x) \le \beta$ almost everywhere.  \revision{We use
    the superscript $+$ to denote time evolution for a discrete-time
    dynamics.}} concepts on neural networks and robust control.

\subsection{Implicit Neural Networks}
We employ a class of deep learning models known as implicit neural
networks (INNs), cf.~\cite{LEG-FG-BT-AA-AT:21}. INNs are based on
implicit prediction rules. For a given input vector $u \in \R^m$, the
output of the INN model is given by
\begin{equation}\label{eq:INN}
  \begin{aligned}
    \hat{y}(u) = Cx + Du + b_y , \,\,
    x = \Phi(Ax + Bu + b_x)
  \end{aligned}
\end{equation}
where $\Phi:\R^n \to \R^n$ is a nonlinear activation map and
$(A,B,C,D,b_x,b_y)$ are the model parameters with appropriate
dimensions.  In~\eqref{eq:INN}, $x \in \R^n$ is the internal state and
is obtained by solving the implicit equation.  In our treatment, we
exploit the fact that the structure of the prediction equations makes
INNs more amenable to control-theoretic analysis.

Any multi layer
perceptron (MLP) can be written in the form
\eqref{eq:INN}. 
Formally, consider a MLP of $L$ hidden layers with input $u \in \R^m$,
output $\hat{y}(u) \in \R^p$, and nonlinearity $\phi: \R \to \R$. \revision{When applying the nonlinearity to vector inputs, we assume that it acts componentwise.} For layer
$l = 0,1,\dots,L-1$, let $x_l \in \R^{n_l}$ be the hidden state and
$\left\{W_l,b_l\right\}$ the weights and biases with appropriate
dimensions. The MLP is given by
\begin{equation}\label{eq:MLP}
  \begin{aligned}    
    x_{l+1} &= \phi(W_l x_l + b_l), \quad l = 0,1,\dots,L-1, \\
    \hat{y}(u) &= W_L x_L + b_L, \,\, x_0 = u 
  \end{aligned}
\end{equation}

\begin{lemma}[{MLP to INN \cite[Sec
    3.2]{LEG-FG-BT-AA-AT:21}}]\label{lemma:MLP-INN}
  The MLP \eqref{eq:MLP} has the following INN representation
  \begin{equation}
    \label{eq:INN-MLP}
    \begin{aligned}
      \hat{y}(u) = H s + b_y, \,\, s = \phi(Fs + Gu + b_x)
    \end{aligned}
  \end{equation}
  where the matrices $F,G,H,b_x,b_y$ are given by
  {\small
    \begin{equation}
      F \hspace{-0.1cm}=\hspace{-0.1cm} \begin{bmatrix}
                                          0 \hspace{-0.1cm}& \hspace{-0.1cm}W_{L - 1} \hspace{-0.1cm}& \hspace{-0.1cm}0 \hspace{-0.1cm}& \hspace{-0.1cm}\dots \hspace{-0.1cm}& \hspace{-0.1cm}0 \\
                                          0 \hspace{-0.1cm}& \hspace{-0.1cm}0 \hspace{-0.1cm}& \hspace{-0.1cm}W_{L - 2} \hspace{-0.1cm}& \hspace{-0.1cm}\dots \hspace{-0.1cm}& \hspace{-0.1cm}0 \\
                                          \vdots \hspace{-0.1cm}& \hspace{-0.1cm}\vdots \hspace{-0.1cm}& \hspace{-0.1cm}\vdots \hspace{-0.1cm}& \hspace{-0.1cm}\vdots \hspace{-0.1cm}& \hspace{-0.1cm}\vdots \\    
                                          0 \hspace{-0.1cm}& \hspace{-0.1cm}0 \hspace{-0.1cm}& \hspace{-0.1cm}0 \hspace{-0.1cm}& \hspace{-0.1cm}\dots \hspace{-0.1cm}& \hspace{-0.1cm}W_1 \\
                                          0 \hspace{-0.1cm}& \hspace{-0.1cm}0 \hspace{-0.1cm}& \hspace{-0.1cm}0 \hspace{-0.1cm}& \hspace{-0.1cm}\dots \hspace{-0.1cm}& \hspace{-0.1cm}0 \\
\end{bmatrix}, \,\, 
G \hspace{-0.1cm}=\hspace{-0.1cm} \begin{bmatrix}
    0 \\
    0 \\
    \vdots \\
    0 \\
    W_0
                                  \end{bmatrix}, \,\,
b_x \hspace{-0.1cm}=\hspace{-0.1cm} \begin{bmatrix}
    b_{L-1} \\
    b_{L-2} \\
    \vdots \\
    b_1 \\
    b_0
\end{bmatrix}
\label{eq:INN_F_G}
\end{equation}
\begin{equation}
H = \begin{bmatrix}
W_L \hspace{-0.1cm}&0 \hspace{-0.1cm}&\dots \hspace{-0.1cm}& 0    
\end{bmatrix}, \,
{ s = \text{vec} [x_L \, x_{L-1} \, \dots \, x_1], \,\, }
b_y = b_L 
\end{equation}
}
\end{lemma}
The well-posedness of the implicit equation~\eqref{eq:INN-MLP}
\revision{(meaning that it has a unique solution) is guaranteed by
  back substitution due to the upper triangular structure of~$F$.}


\subsection{Quadratic Constraints and Neural Network Activations}
We use local incremental quadratic constraints to represent neural
network nonlinearities.  Given
$\mathcal{X},\mathcal{Y} \subseteq \R^n$ closed and
$\mathcal{Q} \subset \mathbb{S}^{2n}$, nonlinearity
$\phi : \R^n \to \R^n$ satisfies the incremental quadratic constraint
($\delta$QC) defined by $(\mathcal{X},\mathcal{Y},\mathcal{Q})$ if
\begin{equation*}
  \begin{bmatrix}
    \phi(x) - \phi(y) \\
    x - y     
  \end{bmatrix}^T 
  Q
  \begin{bmatrix}
    \phi(x) - \phi(y) \\
    x - y  
  \end{bmatrix} \geq 0,  \quad \forall (x,y) \in \mathcal{X}
  \times \mathcal{Y},
\end{equation*}
for any $Q \in \mathcal{Q}$.  Quadratic Constraints (QCs) are defined
analogously by omitting the argument~$y$. The next result
characterizes the $\delta$QCs satisfied by neural network activations.

\begin{lemma}\longthmtitle{Neural Network
    Nonlinearities~\cite{MF-AR-HH-MM-GP:19}}\label{lemma:NN_nonlinearities}
  Let $\xi:\R \to \R$ be slope-restricted on $[\alpha,\beta]$, where
  $\alpha,\beta \in \R$.  Let
  $\mathcal{T}_n := \left\{\, T \in \R^{n \times n} \mid T =
    \diag_n(\lambda) \,, \,\, \lambda \in \R_{>0} \right\}$.
  %
  %
  Then, the vector-valued function
  $\phi(x) = [\xi(x_1) \, \dots \, \xi(x_n)]^T: \R^n \to \R^n$
  satisfies the $\delta$QC defined by $(\R^n,\R^n,\mathcal{Q})$, where
  \begin{equation*}
    \mathcal{Q} = \Big\{ Q \in \mathbb{R}^{2n \times 2n} \, \Bigg| \, 
    Q = \hspace{-0.1cm}
      \begin{bmatrix}
        -2T \hspace{-0.2cm}& (\alpha + \beta) T \\ 
        (\alpha + \beta) T \hspace{-0.2cm}& -2\alpha \beta T
      \end{bmatrix} , \, \forall \, T \in \mathcal{T}_n \Big\} . 
  \end{equation*}
\end{lemma}
\smallskip

This type of representation
allows us to derive LMIs for controller synthesis for systems with
neural feedback loops, at the expense of introducing conservatism.

\subsection{Robust Control}

Linear fractional representations (LFR) are a general, powerful tool
to represent uncertainty in dynamical systems \cite{KZ-JCD:98}.
Consider an autonomous system $x^+ = F(\delta) x$ with uncertainty,
where $F: \R^p \to \R^{n\times n}$ is a matrix-valued function that
depends rationally on the uncertain parameter $\delta \in \R^p$. An LFR of $F(\delta)$ is a pair $(H,\Delta(\delta))$, where
\begin{equation*}
  H =
  \begin{bmatrix}
    A & B \\
    C & D
  \end{bmatrix} \in \R^{(n + n_z) \times (n + n_w)}
\end{equation*}
does not depend on $\delta$ and $\Delta: \R^p \to \R^{n_w \times n_z}$
is a linear function of $\delta$ that satisfies the following: for all
$\delta$ for which $I - D \, \Delta(\delta)$ is invertible and for all
$(\eta,\zeta) \in \R^{2n}$,
%
%
it holds that $\zeta = F(\delta) \eta$ iff there exist
$w \in \R^{n_w}$ and $z \in \R^{n_z}$ such that
\begin{equation*}
  \begin{bmatrix}
    \zeta \\
    z
  \end{bmatrix} = 
  \begin{bmatrix}
    A & B \\
    C & D
  \end{bmatrix}
  \begin{bmatrix}
    \eta \\
    w
  \end{bmatrix}, \quad w = \Delta(\delta)z          
\end{equation*}
The LFR is well posed at $\delta$ if $I - D \, \Delta(\delta)$ is
invertible. We refer to $\Delta$ as the matrix parameter. It can be
shown~\cite{CS-SW:00} that any uncertain dynamical system with rational parameter dependence can be rewritten as an LFR,
\begin{equation*}
  \begin{bmatrix}
    x^+ \\
    z
  \end{bmatrix} = 
  \begin{bmatrix}
    A & B \\
    C & D
  \end{bmatrix}
  \begin{bmatrix}
    x \\
    w
  \end{bmatrix}, \quad w = \Delta(\delta)z
\end{equation*}
In what follows, we drop the $\delta$ dependence for simplicity.

%
%
%
%

Given an LFR depending on the matrix parameter
$\Delta \in \mathbf{\Delta} \subset \R^{n_{w_u} \times n_{z_u}}$,
%
%
consider an additional performance channel
$w_p \in \R^{n_{w_p}},z_p \in \R^{n_{z_p}}$ modeling variables of
interest
%
%
%
to obtain
\begin{equation}\label{eq:LFR2}
  \begin{bmatrix}
    x^+ \\
    z_u \\
    z_p
  \end{bmatrix} = 
  \begin{bmatrix}
    A & B_u & B_p \\
    C_u & D_{uu} & D_{up} \\
    C_p & D_{pu} & D_{pp}
  \end{bmatrix}
  \begin{bmatrix}
    x \\
    w_u \\
    w_p
  \end{bmatrix}, \quad w_u = \Delta z_u
\end{equation}
%
Suppose the LFR is well posed for all $\Delta \in
\mathbf{\Delta}$. Then, \textbf{Robust quadratic performance} is achieved with
index $\Pi_p =
\begin{bsmallmatrix}
  Q_p & S_p \\
  S_p^T & R_p
\end{bsmallmatrix} $ conformable to
$\begin{bsmallmatrix}
   w_p \\ z_p
 \end{bsmallmatrix}$ with $R_p \succeq 0$, $Q_p \prec 0$ if:
 \begin{LaTeXdescription}
\item[Robust stability:] there exists constants $K,\alpha \in \R$ such
  that, 
  \begin{equation*}
    \|x(t_1)\| \leq K \, \|x(t_0)\|^{-\alpha (t_1 - t_0)} \,\,
    \text{for $t_1 \geq t_0$} 
  \end{equation*}
  if $w_p(t) = 0$ and for all $\Delta \in \mathbf{\Delta}$; and
\item[Dissipativity:] there exist $\epsilon,\delta > 0$
  %
  %
  such that
  \begin{equation*}
    \sum_{k = 0}^\infty
    \begin{bmatrix}
      w_p(k)
      \\
      z_p(k)
    \end{bmatrix}^T
    \, \Pi_p
    \,
    \begin{bmatrix}
      w_p(k)
      \\
      z_p(k)
    \end{bmatrix}
    \leq - \epsilon
    \sum_{k = 0}^\infty \|w_p(k)\|^2 
  \end{equation*}
  for $\|w_p\|^2 \leq \delta$, $x(0) = 0$ and all $\Delta \in \mathbf{\Delta}$.
\end{LaTeXdescription}
Without the addition of the performance channel, the notion reduces to
robust stability.  Robust quadratic performance can be checked via
LMIs, cf.\cite[Thm~2]{CWS:01}, whose structure remains the same if the
performance channel is dropped. Hence, in our discussion we only deal
with the robust stability problem.
%
%

%
%
%
%

\section{Problem Formulation}\label{sec:problem-formulation}
Consider the following discrete-time generalized bilinear system with
state $z \in \R^l$ and control input $u \in \R^m$
\begin{equation}\label{eq:og_sys}
  z^+ = A_0 z + B_0 u + \revision{\tilde{D}\, (z \otimes u) + \phi(u)}
  + \Psi(u) z 
\end{equation}
where $A_0 \in \R^{l \times l}$, $B_0 \in \R^{l \times m}$,
$\tilde{D} \in \R^{m \times lm}$, $\phi: \R^m \to \R^l$ and
$\Psi: \R^m \to \R^{l \times l}$ are implicit neural networks (INNs)
that approximate certain nonlinear functions, \revision{with their weights pre-trained and fixed.} Note that $\phi,\Psi$ are in general nonlinear functions of $u$.  We refer to \eqref{eq:og_sys}
as a \emph{bilinear-NFL} system.  Such systems arise naturally when
learning unknown non-affine dynamics from data. \revision{For instance, using the Koopman control family framework~\cite{MH-JC:24-auto}, and given a general
  nonlinear system $x^+ = f(x,u)$, one can obtain an approximate
  lifted representation of the form $z^+ = \mathcal{A}(u) z$, where
  $\mathcal{A}$ is a matrix-valued function parameterized by a
  specific function family (e.g., neural networks), and
  $z = \text{vec}[H(x),1]$ is the lifted state. The bilinear NFL
  \eqref{eq:og_sys} arises naturally by decomposing $\mathcal{A}$ into
  a linear term and a residual nonlinearity.}
%
%

%

Assume that the system~\eqref{eq:og_sys} is stabilizable. Let $z_*$ be
an equilibrium when $u = u_*$.  We also assume the same neural network
nonlinearity in $\phi,\Psi$ and that it acts elementwise and is
slope-restricted on $[\alpha,\beta]$.
%
%
%
%
Let, $Q_z \prec 0$ and $R_z \succ 0$. Consider the region of interest $\mathcal{Z}$
around $z_*$ given by
\begin{equation}\label{eq:Z_defn}
  \mathcal{Z} = z_* + \Big\{z \in \R^l \, \Bigg| \, 
\begin{bmatrix}
  z \\ 
  1
\end{bmatrix}^T \, 
\begin{bmatrix}
  Q_z & S_z \\ 
  S_z^T & R_z
\end{bmatrix} \, 
\begin{bmatrix}
  z \\ 
  1
\end{bmatrix} \, 
\geq 0
\Big\},
\end{equation}

\textbf{Problem:} Design a static state feedback controller
$u = k(z)$, with $u_* = k(z_*)$, to enforce local exponential
stability of $z_*$ while ensuring that the system state $z$ remains
in~$\mathcal{Z}$.

This problem is challenging due to control entering nonlinearly in the
dynamics~\eqref{eq:og_sys} through the nonlinear function $\phi$, the
lack of a closed-form expression for the latter, and the bilinearity
of the term $\Psi(u)z$.
To tackle it, we adopt a robust-control inspired approach.  We start
by abstracting out the nonlinearities of the system, representing them
using uncertain matrices. The basic idea relies on the observation
that a nonlinear function $f \in \mathcal{X} \subseteq \R^n \to \R^n$
can be represented as $f(x) = \Delta_{f}(x) x$, where
$\Delta_f(x) \in \R^{n \times n}$. Instead of dealing the explicit
dependence of $\Delta_f$ on $x$, we treat such matrices as
uncertainties that can be measured online.  Next, we construct an
appropriate uncertainty set, denoted
$\mathbf{\Delta_f} \subseteq \R^{n \times n}$,
and require it to be such that
$\mathbf{\Delta_f} \supseteq \{\Delta \in \R^{n \times n} \, \vert \,
f(x) = \Delta \, x, x \in \mathcal{X} \}$.
%
%
The resulting construction is then amenable to controller synthesis
via LMIs.  This over-approximation introduces conservatism, which is
the cost for obtaining a convex formulation for controller synthesis.

\revision{
  Although we work with \eqref{eq:og_sys} for simplicity of
  exposition, we note that the controller design procedure presented
  here also works for the more general class of systems
    \begin{align}\label{eq:general_sys}
      z^+ &= A_0 z + B_0 u + \tilde{D}\, (z \otimes u) +  \phi_u(u)
            \notag
      \\ 
          & \quad + \phi_z(z) + \Psi_z(z) u + \Psi_u(u) z + \Psi(u,z),
    \end{align}      
    where $\phi_u,\phi_z,\Psi_u,\Psi_z,\Psi$ are neural networks of
    appropriate sizes. 
    In this model, the presence of the terms beyond $\Psi$ strikes a
    balance between specificity and generality to exploit structure
    and reduce conservatism in controller design.
}

\section{Controller Design}\label{sec:Control_design}
Here we carry out the design strategy outlined in
Section~\ref{sec:problem-formulation}. We reformulate the bilinear
NFL~\eqref{eq:og_sys}
%
%
to obtain an LFR by abstracting out all nonlinearities in the
system using fabricated uncertainties, cf. Section
\ref{sec:reform_INN}.  Next, we obtain appropriate uncertainty sets
for representing these uncertainties using quadratic constraints,
cf. Section \ref{sec:unc_set_desc}. This gives rise to a reformulation
of the nonlinear control problem as a linear parameter varying (LPV)
control problem, for which we obtain a set of LMIs for controller
synthesis, cf. Section \ref{sec:con_syn}.

\subsection{LFR Reformulation}\label{sec:reform_INN}
\vspace{-0.14cm}
%
%
Let $\{\psi_{i}: \R^m \to \R^l\}_{i=1}^l$ denote the columns of the
INN $\Psi: \R^m \to \R^{l \times l}$ in~\eqref{eq:og_sys}, so that
$ \Psi(u) = [ \psi_{1}(u) \; \psi_{2}(u) \; \dots \; \psi_{l}(u)]$,
for $u \in \R^m$.  Each $\psi_{i}$ and $\phi$ are represented using
well-posed INNs of the form \eqref{eq:INN} as
\begin{align}\label{eq:phi_implicit}
  {
  \phi(u)}
  & { = \!
    H^{\phi} s^{\phi} \!+\! J^{\phi} u \!+\! b_y^{\phi}, \, s^{\phi} =
    \xi(F^{\phi} {s}^{\phi} \!+\! G^{\phi} u \!+\! {b}_x^{\phi}) ,} \notag
  \\
  \psi_{{i}}(u)
  & = \! H_{{i}}^{\psi} {s}_{i}^{\psi} \!+\!
    J_{i}^{\psi} u \!+\! {b}_y^i, \, {s}_{{i}}^{\psi}
    = \xi(F_{i}^{\psi} {s}_{i}^{\psi} \!+\!
    G_{i}^{\psi} u \!+\! {b}_x^i)
\end{align}
\revision{
  $F^{\#}_i \in \R^{k_{\#} \times k_{\#}},G^{\#}_i \in \R^{k_{\#}
    \times m},H^{\#}_i \in \R^{l \times k_{\#}},J^{\#}_i \in \R^{l
    \times m}$ for $\# \in \{\psi,\phi\}$ are constant matrices,
  ${b}^i_x \in \R^{k_{\psi}}, {b}^i_y \in \R^l,{b}^{\phi}_x \in
  \R^{k_{\phi}},{b}^{\phi}_y \in \R^l$ are bias terms} and
$\xi : \R \to \R$ is the neural network nonlinearity that acts
elementwise and is slope-restricted on $[\alpha,\beta]$.
Then, 
\begin{equation}
  \Psi(u)z = \sum_{i=1}^{l}H^{\psi}_i {s}^{\psi}_i z_i + 
  \sum_{i=1}^{l}J^{\psi}_i u z_i + 
  \sum_{i=1}^{l} {b}^i_y z_i. 
  \label{eq:phi_z}
\end{equation}
%
%
\revision{Note that $\Psi(u)z$ is non-affine in $u$ because
  $s^{\psi}_i$ depends nonlinearly on $u$ through the implicit
  equation \eqref{eq:phi_implicit}.}  We introduce
$F^{\psi} \in \R^{l k_{\psi} \times l k_{\psi}}, G^{\psi} \in \R^{l
  k_{\psi} \times m}, H^{\psi} \in \R^{l \times l^2{k_{\psi}}},
s^{\psi} \in \R^{l k_{\psi}}, B^{\psi}_y \in \R^{l \times l}$, and
$b^{\psi}_x \in \R^{l k_{\psi}}$
\begin{align*}
  H^{\psi} & = [ H^{\psi}_1 \, \dots \, H^{\psi}_l] [ E_1 \otimes I_{k_{\psi}} \,  \dots \, E_l
      \otimes I_{k_{\psi}}] \in \R^{l \times l^2{k_{\psi}}},
  \\
  J^{\psi}  & = [J^{\psi}_1 \dots  J^{\psi}_l],
       \,\, B^{\psi}_y  =  [b_y^1  \dots  b_y^l ], \,\,
       F^{\psi} =  \diag [F^{\psi}_1 \dots  F^{\psi}_l ],
  \\
  G^{\psi}  & =\hspace{-0.1cm} \text{vec}[G^{\psi}_1  \dots  G^{\psi}_l] , \,
  s^{\psi} = \hspace{-0.1cm}\text{vec}[ s^{\psi}_1  \dots s^{\psi}_l ], \,
  b^{\psi}_x  = \text{vec}[ b_x^1  \dots b_x^l ].
\end{align*}
%
%
We now show the bilinear-NFL system can be rewritten as an LFR by treating the nonlinearities as uncertainties.
\begin{proposition}\longthmtitle{LFR
    Reformulation}\label{thm:LFR_reformulation}
  The bilinear-NFL system~\eqref{eq:og_sys} can be rewritten as an LFR
  with a bias term
  \begin{subequations}\label{eq:compact_I-SS-all}
    \begin{small}
    \revision{
    \begin{equation}\label{eq:compact_I-SS}
      \hspace{-0.3cm}\begin{bmatrix}
        z^+ \\
        u \\
        s^{\psi} \\
        v^{\phi} \\
        v^{\psi}
      \end{bmatrix}
      \hspace{-0.14cm} = \hspace{-0.14cm}
      \begin{bmatrix}
        \hspace{-0.1cm}A_0 + B^{\psi}_y \hspace{-0.3cm}&  B_0 + J^{\phi} \hspace{-0.3cm}& \tilde{D} + J^{\psi} \hspace{-0.3cm}& H^{\psi} \hspace{-0.3cm}& H^{\phi} \hspace{-0.3cm}& 0 \\
        0 \hspace{-0.3cm}& I \hspace{-0.3cm}& 0 \hspace{-0.3cm}& 0 \hspace{-0.3cm}& 0 \hspace{-0.3cm}& 0 \\
        0 \hspace{-0.3cm}& 0 \hspace{-0.3cm}& 0 \hspace{-0.3cm}& 0 \hspace{-0.3cm}& 0 \hspace{-0.3cm}& I \\
        0 \hspace{-0.3cm}& G^{\phi} \hspace{-0.3cm}& 0 \hspace{-0.3cm}& 0 \hspace{-0.3cm}& F^{\phi} \hspace{-0.3cm}& 0 \\
        0 \hspace{-0.3cm}& G^{\psi} \hspace{-0.3cm}& 0 \hspace{-0.3cm}& 0 \hspace{-0.3cm}& 0 \hspace{-0.3cm}& F^{\psi} \\
      \end{bmatrix} \hspace{-0.2cm}
      \begin{bmatrix}
        z \\
        u \\
        w_u \\
        w_{\psi} \\
        s^{\phi} \\
        s^{\psi}
      \end{bmatrix} \hspace{-0.14cm}+ \hspace{-0.14cm}
      \begin{bmatrix}
        b^{\phi}_y \\
        0 \\
        0 \\
        0 \\
        b^{\phi}_x \\
        b^{\psi}_x 
      \end{bmatrix}
    \end{equation}
    }
  \end{small}    
  with the outputs $u \in \R^m$,
  $s^{\psi} \in \R^{lk_{\psi}},v^{\phi} \in \R^{k_{\phi}}$,
  $v^{\psi} \in \R^{lk_{\psi}}$ nonlinearly fedback using uncertain
  matrices $\Delta_{m} \in \R^{lm \times m}$,
  $\Delta_{k_{\psi}} \in \R^{lk_{\psi} \times k_{\psi}}$,
  $\Delta^{\phi}_{\xi} \in \R^{k_{\phi} \times k_{\phi}}$,
  $\Delta^{\psi}_{\xi} \in \R^{lk_{\psi} \times lk_{\psi}}$, to obtain
  $w_u \in \R^{lm}$, $w_{\psi} \in \R^{l^2k_{\psi}}$,
  $s^{\phi} \in \R^{k_{\phi}}$, $s^{\psi} \in \R^{l k_{\psi}}$ as
  \begin{equation}\label{eq:LFR_fedback}
    \begin{bmatrix}
      w_u \\
      \revision{w_{\psi}} \\
        \revision{s^{\phi}} \\
        \revision{s^{\psi}}
    \end{bmatrix}
    = 
    \begin{bmatrix}
      (z \otimes I_m) u\\
      \revision{\diag_l (z \otimes I_{k_{\psi}}) s^{\psi}} \\
      \revision{\xi(v^{\phi})} \\
      \revision{\xi(v^{\psi})}      
    \end{bmatrix} 
    = 
    \begin{bmatrix}
      \Delta_m \, u \\
      \revision{\diag_l(\Delta_{k_{\psi}}) \, s^{\psi}} \\
      \revision{\Delta^{\phi}_{\xi} \, v^{\phi}} \\
      \revision{\Delta^{\psi}_{\xi} \, v^{\psi}}            
    \end{bmatrix} . 
  \end{equation}
  \end{subequations}
\end{proposition}
%
%
\begin{proof} 
  Note that    
  \vspace*{-3ex}
  \begin{small}    
  \begin{align*}
    \sum_{i = 1}^{l}H^{\psi}_i {s}^{\psi}_i z_i + J^{\psi}_i u z_i\hspace{-0.08cm}=\hspace{-0.08cm} 
    [ H^{\psi}_1 \,  \dots \, H^{\psi}_l] 
    \left[\begin{matrix}
      s^{\psi}_1 z_1 \vspace{-0.1cm}
      \\
      \vdots
      \\
      s^{\psi}_l z_l
    \end{matrix}\right]
    \hspace{-0.08cm}+\hspace{-0.08cm}
    [ J^{\psi}_1 \, \dots \, J^{\psi}_l] 
    \left[\begin{matrix}
      u z_1
      \vspace{-0.1cm}
      \\
      \vdots
      \\
      u z_l
    \end{matrix}\right]
  \end{align*}
  \end{small}
  We can rewrite
  $
  \text{vec}[uz_1 \dots uz_l] 
  = (z \otimes u) = (z \otimes I_m)u$
  and
  \begin{small}
  \begin{multline*}
    \begin{bmatrix}
      s^{\psi}_1 z_1
      \vspace{-0.1cm}
      \\
      \vdots
      \\
      s^{\psi}_l z_l
    \end{bmatrix}
    =
    \begin{bmatrix}
      s^{\psi}_1 z_1
      \vspace{-0.1cm}
      \\
      \vdots
      \\
      0
    \end{bmatrix} + 
    \dots +
    \begin{bmatrix}
      0 \vspace{-0.1cm}
      \\
      \vdots\\
      s^{\psi}_l z_l
    \end{bmatrix} 
    = \sum_{i=1}^{l} (E_i z) \otimes s^{\psi}_i
    \\
    = \sum_{i = 1}^{l}(E_i \otimes I_{k_{\psi}})  (z \otimes s^{\psi}_i) 
    = \sum_{i = 1}^{l}(E_i \otimes I_{k_{\psi}}) (z \otimes I_k) s^{\psi}_i .
  \end{multline*}
\end{small}
  Therefore, we can rewrite \eqref{eq:phi_z} as
  \begin{equation}\label{eq:Psi_expr}
    \Psi(u) z = H^{\psi}
    \begin{bmatrix}
      (z \otimes I_{k_{\psi}}) s^{\psi}_1
      \vspace{-0.1cm}
      \\
      \vdots
      \vspace{-0.1cm}
      \\
      (z \otimes I_{k_{\psi}})s^{\psi}_l
    \end{bmatrix} + J^{\psi} (z \otimes I_m) u + B^{\psi}_y z .
  \end{equation}
  \revision{ Using \eqref{eq:phi_implicit} and the
    expression~\eqref{eq:Psi_expr}, we rewrite \eqref{eq:og_sys} as
    $ z^+ = (A_0 + B^{\psi}_y)z + (B_0 + J^{\phi})u + (\tilde{D} +
    J^{\psi}) w_u + H^{\psi} w_{\psi} + H^{\phi} s^{\phi}$, leading to
    the LFR \eqref{eq:compact_I-SS-all}.}
\end{proof}
%
%

It is important to note the special form of the nonlinear feedback
terms in \eqref{eq:LFR_fedback}, where $\Delta_m$, $\Delta_{k_{\psi}}$
depend only on $z$, which can be measured online and is restricted to
belong to~$\mathcal{Z}$. This is more advantageous than having the
matrices depend on $u$ and/or~$s^{\psi}$. We conclude this section by
shifting the origin to the desired equilibrium $(z_*,u_*)$ in the LFR
representation.

\begin{lemma}\longthmtitle{Shifted LFR}
  Consider the bilinear-NFL system~\eqref{eq:og_sys} and let $z_*$ be
  an equilibrium when $u=u_*$.  Given the LFR
  representation~\eqref{eq:compact_I-SS-all}, let
  $s^{\psi}_* \in \R^{lk_{\psi}}$ and $s^{\phi}_* \in \R^{k_{\phi}}$
  be the corresponding internal neural network states. Then, with the
  shifted variables $\tilde{z} = z - z_*$, $\tilde{u} = u - u_*$,
  $\tilde{s}^{\psi} = s^{\psi} - s^{\psi}_*$, and
  $\tilde{s}^{\phi} = s^{\phi} - s^{\phi}_*$, the LFR representation
  takes the form
  \begin{subequations}\label{eq:shifted_LFR-all}
  \begin{equation}
    \hspace{-0.3cm}\begin{bmatrix}
      \tilde{z}^+ \\
      \tilde{u} \\
      \tilde{s}^{\psi} \\
      \tilde{v}^{\phi} \\
      \tilde{v}^{\psi}
    \end{bmatrix}
    \hspace{-0.1cm} = \hspace{-0.1cm}
    \begin{bmatrix}
      \mathcal{A} & \mathcal{B} & \tilde{D} + J^{\psi} & H^{\psi} & H^{\phi} & H^{\psi}_*\\
      0 & I & 0 & 0 & 0 & 0\\
      0 & 0 & 0 & 0 & 0 & I\\
      0 & G^{\phi} & 0 & 0 & F^{\phi} & 0 \\
      0 & G^{\psi} & 0 & 0 & F^{\psi} & 0  
    \end{bmatrix}    
    \begin{bmatrix}
      \tilde{z} \\
      \tilde{u} \\
      \tilde{w}_u \\
      \tilde{w}_{\psi} \\
      \tilde{s}^{\phi} \\
      \tilde{s}^{\psi} 
    \end{bmatrix} ,
    \label{eq:shifted_LFR}
  \end{equation}
  with the outputs nonlinearly fed-back as
  \begin{equation}\label{eq:shifted_LFR-fedback}
    \begin{bmatrix}
      \tilde{w}_u \\
      \tilde{w}_{\psi} \\
      \tilde{s}^{\phi} \\
      \tilde{s}^{\psi} 
    \end{bmatrix}
    = 
    \begin{bmatrix}
      (\tilde{z} \otimes I_m) \tilde{u}
      \\
      \diag_l (\tilde{z} \otimes I_{k_{\psi}}) \tilde{s}^{\psi}
      \\
      {\beta}^{\phi} \,  (\tilde{v}^{\phi}) \\
      {\beta}^{\psi} \,  (\tilde{v}^{\psi})
    \end{bmatrix}   
    =
    \begin{bmatrix}
      \Delta_m \, \tilde{u} \\
      \diag_l(\Delta_{k_{\psi}}) \, \tilde{s}^{\psi} \\
      \Delta_{\phi} \, \tilde{v}^{\phi} \\
      \Delta_{\psi} \, \tilde{v}^{\psi}       
    \end{bmatrix} ,
  \end{equation}
\end{subequations}
where $v^{\phi}_* = F^{\phi} s^{\phi}_* + G^{\phi} u_* + b^{\phi}_x$,
$v^{\psi}_* = F^{\psi} s^{\psi}_* + G^{\psi} u_* + b^{\psi}_x$,
${\beta}^{\phi}(\tilde{v}^{\phi}) = \xi(\tilde{v}^{\phi} + v^{\phi}_*)
- \xi(v^{\phi}_*)$, and
${\beta}^{\psi}(\tilde{v}^{\psi}) = \xi(\tilde{v}^{\psi} + v^{\psi}_*)
- \xi(v^{\psi}_*)$ are the shifted neural network
nonlinearities. Here,
$\mathcal{A} = A_0 + B^{\psi}_y + A^*_u + A^*_{s^\psi}$,
$\mathcal{B} = B_0 + J^{\phi} + (\tilde{D} + J^{\psi}) (z_* \otimes
I_m)$,
$H^{\psi}_* = H^{\psi} \text{diag}_l(z_* \otimes I_{k_{\psi}})$, and
  \begin{align*}
    A^*_{s^\psi}
    &= H^{\psi} [
      \diag_l(e_1 \otimes I_{k_{\psi}}) s^{\psi}_* \, \dots \, \diag_l(e_l \otimes I_{k_{\psi}}) s^{\psi}_* 
      ] ,
    \\
    A^*_u
    &= (\tilde{D} + J^{\psi})[
      (e_1 \otimes I_m) u_* \,   \dots \, (e_l \otimes
      I_m)u_* ] . \eqoprocend
  \end{align*} 
\end{lemma}

The proof readily follows by shifting the origin. In what follows, we
drop the tilde arguments for ease of notation. 

\subsection{Uncertainty Set Description}\label{sec:unc_set_desc}

Here, we complement the LFR formulation~\eqref{eq:shifted_LFR-all} by
identifying uncertainty sets $\mathbf{\Delta_m}$,
$\mathbf{\Delta_{k_{\psi}}}$, $\mathbf{\Delta_{\phi}}$, and
$\mathbf{\Delta_\psi}$ for the matrices $\Delta_m $,
$\Delta_{k_{\psi}}$, $\Delta_{\phi}$, and $\Delta_{\psi}$,
respectively, using quadratic constraints.  The matrices
$\Delta_{\phi} \in \mathbf{\Delta_{\phi}}$ and
$\Delta_{\psi} \in \mathbf{\Delta_{\psi}}$ represent the neural
network nonlinearities and share the same structure. Similarly,
matrices $\Delta_m \in \mathbf{\Delta_m}$ and
$\Delta_{k_{\psi}} \in \mathbf{\Delta_{k_{\psi}}}$ represent the
bilinearities and have the same structure but different sizes.

Using Lemma~\ref{lemma:NN_nonlinearities}, for all
$T \in \mathcal{T}_{l{k_{\psi}}}$ the shifted nonlinearity
${\beta}^{\psi}({v}^{\psi}) = \xi({v}^{\psi} + v^{\psi}_*) -
\xi(v^{\psi}_*)$ satisfies
%
\begin{equation}
  \begin{bmatrix}
    \star
  \end{bmatrix}^T 
  \begin{bmatrix} 
    -2T & (\alpha + \beta) T \\ 
    (\alpha + \beta) T & -2\alpha \beta T
  \end{bmatrix} 
  \begin{bmatrix} 
    {\beta}^{\psi}({v}^{\psi}) \\
    v^{\psi}
  \end{bmatrix} \geq 0  , \; \forall v^{\psi} \in  \R^{l k_{\psi}}.
\end{equation}
Therefore, we define the neural network uncertainty set as
  \begin{align} \label{eq:Delta_psi_eqn}
    \mathbf{\Delta}_{\psi}
    &= \Big\{
      \Delta_{\psi} \in \R^{lk_{\psi} \times lk_{\psi}} \, \Big| \,\, \forall \,\, T
      \in \mathcal{T}_{lk_{\psi}} , \notag \\  
    &\begin{bmatrix}
       \Delta_{\psi} \\
       I
     \end{bmatrix}^T 
      \begin{bmatrix}
        -2T & (\alpha + \beta)T \\
        (\alpha + \beta)T & -2 \alpha \beta T
      \end{bmatrix}
      \begin{bmatrix}
        \Delta_{\psi} \\
        I
      \end{bmatrix} \geq 0    
      \Big \}. 
  \end{align}  
$\mathbf{\Delta_{\phi}} \subset \R^{k_{\phi} \times k_{\phi}}$ is defined similarly. 
To tackle the bilinear terms, we rely on the
following result.

\begin{lemma}[Bilinearity as an Uncertainty~\cite{RS-JB-FA:23}]\label{le:bilinear-uncertainty}
  Let $\mathcal{Z}$ as defined in~\eqref{eq:Z_defn}.  For any
  $m \in \mathbb{Z}$, define the uncertainty set {\small
  \begin{align} \label{eq:Delta_m} 
    \mathbf{\Delta}_m
    &= \Big\{\Delta \in \R^{m l \times m} \, \Big| \,\, \forall \,\,
      0 \preceq \Lambda_m \in \R^{m \times m}, \notag
    \\
    &\quad
      \begin{bmatrix}
        \Delta \\ 
        I
      \end{bmatrix}^T 
      \begin{bmatrix}
        Q_z \otimes \Lambda_m & S_z \otimes \Lambda_m \\ 
        S_z^T \otimes \Lambda_m & R_z \otimes \Lambda_m
      \end{bmatrix} 
      \begin{bmatrix}
        \Delta \\ 
        I
      \end{bmatrix} \, \succeq 0
      \Big\}.        
  \end{align}         
  }
  Then, $\Delta \in \mathbf{\Delta}_m$ iff
  $\exists \, z \in \mathcal{Z}$ such that $\Delta = z \otimes I_m$. \oprocend
  %
  %
\end{lemma}
%
%

This result allows us to define the uncertainty sets
$\mathbf{\Delta_m}$ and $\mathbf{\Delta_{k_{\psi}}}$ for the
bilinearities.  Note that Lemma~\ref{le:bilinear-uncertainty} relies
on the fact that $z \in \mathcal{Z}$. This means that, in our
forthcoming controller design, we need to ensure
\revision{$z(t) \in \mathcal{Z}$ for all time~$t$}.
%
%

To conclude, we combine the uncertain matrices into
${\Delta}_{c} = \diag(\Delta_m,\diag_l(\Delta_{k_{\psi}}), \Delta_{\phi}, \Delta_{\psi})$ to
obtain a compact representation of the
uncertainties~\eqref{eq:shifted_LFR-fedback} fed-back in the LFR. Let
$m_c = lm + l^2k_{\psi} + lk_{\psi} + k_{\phi} , \, n_c = 2lk_{\psi} + m + k_{\phi}$ and define 
$q = \text{vec}(w_u,w_{\psi},s^{\phi},s^{\psi}) \in \R^{m_c}$, $p = \text{vec}(u,s^{\phi},v^{\phi},v^{\psi}) \in \R^{n_c}$. Then, we have $q = \Delta_{c} p, \text{ where } \Delta_{c} \in \mathbf{\Delta}_{c}$ is the combined uncertainty set given by
%
%
\begin{equation}
  \hspace{-0.1cm}
  \begin{aligned}
    &\mathbf{\Delta}_{c}
    = 
      \bigg\{
      \Delta_c \in \R^{m_c \times n_c} \, \bigg| \, \forall \,
      \Lambda_m \succeq 0, \, \Lambda_{k_{\psi}} \succeq 0, \, T_{k_{\phi}} \in
      \mathcal{T}_{k_{\phi}},  
    \\
    &
    T_{l k_{\psi}} \in
      \mathcal{T}_{l k_{\psi}}, \,
      \begin{bmatrix}
        \Delta_c
        \\ 
        I_{n_c}
      \end{bmatrix}^T 
      \begin{bmatrix}
        {Q} & {S}
        \\ 
        S^T & {R}
      \end{bmatrix} 
      \begin{bmatrix}
        \Delta_c
        \\ 
        I_{n_c}
      \end{bmatrix} \, \succeq 0 
      \bigg\}.
  \end{aligned}
  \label{eq:Delta_combined}
\end{equation}
The matrices $Q \in \R^{m_c \times m_c}$, $S \in \R^{m_c \times n_c}$, $R \in \R^{n_c \times n_c}$ 
\revision{
  \begin{small}    
    \begin{equation*}
      \begin{aligned}
        Q 
        &= \diag 
          \left[
          \begin{matrix}
            Q_z \otimes \Lambda_m & \hspace{-0.2cm} \diag_l (Q_z \otimes \Lambda_{k_{\psi}}) &\hspace{-0.1cm}  -2 T_{k_{\phi}} &\hspace{-0.1cm} -2 T_{l k_{\psi}}
        \end{matrix}
        \right]    \\
    S
      &= \diag 
        \left[
        \begin{matrix}
        S_z \otimes \Lambda_m &\hspace{-0.2cm}  \diag_l (S_z \otimes \Lambda_{k_{\psi}}) &\hspace{-0.1cm} (\alpha \hspace{-0.07cm}+\hspace{-0.07cm} \beta) T_{k_{\phi}} &\hspace{-0.1cm} (\alpha \hspace{-0.07cm}+\hspace{-0.07cm} \beta) T_{l k_{\psi}}
        \end{matrix}
        \right]    \\
    R 
    &= \diag 
      \left[
      \begin{matrix}
      R_z \otimes \Lambda_m &\hspace{-0.2cm}  \diag_l (R_z \otimes \Lambda_{k_{\psi}}) &\hspace{-0.1cm} -2\alpha \beta T_{k_{\phi}} &\hspace{-0.1cm} -2\alpha \beta T_{l k_{\psi}}
      \end{matrix}
      \right].                
  \end{aligned}
\end{equation*}
  \end{small}
}
%
%
The next result computes the multiplier matrix inverse in the
definition of~$\mathbf{\Delta}_c$, as it is useful later in our
controller design.
\revision{The proof is omitted for space reasons.}

\begin{lemma}\longthmtitle{Multiplier Matrix
    Inverse in Combined Uncertainty Set}
  Given the set $\mathcal{Z}$ in \eqref{eq:Z_defn} and the combined
  uncertainty set $\mathbf{\Delta}_c$ in~\eqref{eq:Delta_combined},
  let
  \begin{equation*}
    \begin{bmatrix}
      \tilde{Q}_z & \tilde{S}_z
      \\ 
      \tilde{S}_z^T & \tilde{R}_z
    \end{bmatrix}
    =
    \begin{bmatrix}
      Q_z & S_z \\ 
      S_z^T & R_z
    \end{bmatrix}^{-1}
    ,
    \quad 
    \begin{bmatrix}
      \tilde{Q} & \tilde{S}
      \\ 
      \tilde{S}^T & \tilde{R}
    \end{bmatrix}
    = 
    \begin{bmatrix}
      Q & S \\ 
      S^T & R
    \end{bmatrix}^{-1}  .
  \end{equation*}
  Then, $\tilde{Q} \in \R^{m_c \times m_c}$, $\tilde{S} \in \R^{m_c \times n_c}$, $\tilde{R} \in \R^{n_c \times n_c}$ are
  \begin{small}    
  \begin{equation*}
    \begin{aligned}
      \tilde{Q} 
      &\hspace{-0.06cm}=\hspace{-0.06cm}\text{diag} 
      \left[
        \begin{matrix}
          \tilde{Q}_z \otimes \tilde{\Lambda}_m &\hspace{-0.2cm} \diag_l (\tilde{Q}_z \otimes \tilde{\Lambda}_{k_{\psi}}) & \hspace{-0.2cm}
          \frac{2 \alpha \beta}{(\alpha \hspace{-0.03cm}-\hspace{-0.03cm} \beta)^2}
          \tilde{T}_{k_{\phi}} &\hspace{-0.2cm} \frac{2 \alpha \beta}{(\alpha \hspace{-0.03cm}-\hspace{-0.03cm} \beta)^2}\tilde{T}_{l k_{\psi}}
        \end{matrix} \right]  
      ,
      \\
      \tilde{S}
      &\hspace{-0.06cm}=\hspace{-0.06cm} \text{diag}
      \left[
        \begin{matrix}
          \tilde{S}_z \otimes \tilde{\Lambda}_m &\hspace{-0.2cm} \diag_l (\tilde{S}_z \otimes \tilde{\Lambda}_{k_{\psi}}) &\hspace{-0.2cm} \frac{\alpha \hspace{-0.03cm}+\hspace{-0.03cm} \beta}{(\alpha \hspace{-0.03cm}-\hspace{-0.03cm} \beta)^2}\tilde{T}_{k_{\phi}} &\hspace{-0.2cm} \frac{\alpha \hspace{-0.03cm}+\hspace{-0.03cm} \beta}{(\alpha \hspace{-0.03cm}-\hspace{-0.03cm} \beta)^2}\tilde{T}_{l k_{\psi}}
        \end{matrix} \right]  
      ,
      \\
      \tilde{R}
      &\hspace{-0.06cm}=\hspace{-0.06cm} \text{diag}
      \left[
      \begin{matrix}
          \tilde{R}_z \otimes \tilde{\Lambda}_m &\hspace{-0.2cm} \diag_l (\tilde{R}_z \otimes \tilde{\Lambda}_{k_{\psi}}) &\hspace{-0.2cm} \frac{2}{(\alpha \hspace{-0.03cm}-\hspace{-0.03cm} \beta)^2}\tilde{T}_{k_{\phi}} &\hspace{-0.2cm} \frac{2}{(\alpha \hspace{-0.03cm}-\hspace{-0.03cm} \beta)^2}\tilde{T}_{l k_{\psi}}
            \end{matrix} \right] ,
    \end{aligned}        
  \end{equation*}
  \end{small}where $\tilde{\Lambda}_m,\tilde{\Lambda}_{k_{\psi}},\tilde{T}_{k_{\phi}},\tilde{T}_{l k_{\psi}}$ are the inverses of ${\Lambda}_m,{\Lambda}_{k_{\psi}},{T}_{k_{\phi}},{T}_{k_{l \psi}}$. \oprocend
\end{lemma}

\subsection{Controller Synthesis}\label{sec:con_syn}

Having successfully reformulated the bilinear-NFL system as an LFR,
with associated uncertainty set descriptions, here we tackle the
controller design. We parameterize our design by a linear combination
of all the inputs to the LFR \eqref{eq:shifted_LFR-all} as,
\begin{align*}
  u = K_z \, z + K_u \, w_u + K_{w_{\psi}} \, w_{\psi} + K_{\phi} \,
  s^{\phi} + K_{\psi} s^{\psi}, 
\end{align*}
where $K_z \in \R^{m \times l}$, $K_u \in \R^{m \times lm}$,
$K_{w_{\psi}} \in \R^{m \times l^2k_{\psi}}$,
$K_{s^{\phi}} \in \R^{m \times k_{\phi}}$, and
$K_{s^{\psi}} \in \R^{m \times lk_{\psi}}$. This leads to
\begin{small}  
  \begin{align}
    \label{eq:controller_implicit}
    &u = K_zz \hspace{-0.08cm}+\hspace{-0.08cm} K_u  (z \hspace{-0.04cm}\otimes\hspace{-0.04cm}
      I_m) u \hspace{-0.08cm}+\hspace{-0.08cm} K_{w_{\psi}} \diag_l(z \hspace{-0.04cm}\otimes\hspace{-0.04cm} I_{k_{\psi}}) s^{\psi} 
      \hspace{-0.08cm}+\hspace{-0.08cm} K_{\phi} s^{\phi}
      \hspace{-0.08cm}+\hspace{-0.08cm} K_{\psi} s^{\psi}, \notag
    \\
    &s^{\#} = \xi(F^{\#} s^{\#} + G^{\#} u + v^{\#}_*) - \xi(v^{\#}_*), \text{where $\# \in \{\phi,\psi \}$}.
  \end{align}
\end{small}
The controller is obtained by solving the implicit equation \eqref{eq:controller_implicit}, is nonlinear in the state and depends on the
neural network weights. Using~\eqref{eq:controller_implicit}, the
closed-loop LFR is
\begin{equation}\label{eq:LFR_closed_loop}
    \begin{bmatrix}
      z^+ \\
      p 
    \end{bmatrix}
    = 
      \begin{bmatrix}
        A & B
        \\
        C & D  
      \end{bmatrix}
      \,
      \begin{bmatrix}
        z \\
        q
      \end{bmatrix} ,\,       
    q = \Delta_{c} \, p , \quad \Delta_{c} \in \mathbf{\Delta}_{c} ,        
\end{equation}
where $A = \mathcal{A} + \mathcal{B}K_z$,
$B = [ (\tilde{D} + J^{\psi}+\mathcal{B} K_u) \,|\, (H^{\psi} + \mathcal{B} K_{w_{\psi}}) \,|\, (H^{\phi} + \mathcal{B} K_{\phi}) \,|\, (H^{\psi} \diag_l{z_* \otimes I_{k_{\psi}}} + \mathcal{B} K^{\psi})] $,
\begin{small}
\begin{equation*}
  C \hspace{-0.1cm}=\hspace{-0.1cm}
  \left[
  \begin{matrix}\hspace{-0.1cm}
    K_z \\ 
    0 \\ 
    G^{\phi} \, K_z \\ 
    G^{\psi} \, K_z
  \hspace{-0.1cm}\end{matrix} 
  \right]\hspace{-0.1cm}, 
  D \hspace{-0.1cm}=\hspace{-0.1cm}
  \left[
  \begin{matrix}\hspace{-0.1cm}
    K_u &\hspace{-0.2cm} K_{w_{\psi}} &\hspace{-0.2cm} K_{\phi} &\hspace{-0.3cm} K_{\psi} \\
    0 &\hspace{-0.2cm} 0 &\hspace{-0.2cm} 0 &\hspace{-0.3cm}I \\
    G^{\phi} K_u &\hspace{-0.2cm} G^{\phi}  K_{w_{\psi}} &\hspace{-0.2cm} F^{\phi} + G^{\phi} K_{\phi} &\hspace{-0.3cm} G^{\phi} K_{\psi} \\
    G^{\psi} K_u &\hspace{-0.2cm} G^{\psi}  K_{w_{\psi}} &\hspace{-0.2cm} G^{\psi} K_{\phi} &\hspace{-0.3cm} F^{\psi} + G^{\psi} K_{\psi}
  \hspace{-0.05cm}\end{matrix}
  \right].
\end{equation*}
\end{small}
The controller synthesis problem is then to find matrices $K_z$,
$K_u$, $K_{w_{\psi}}$, $K_{\phi}$, and $K_{\psi}$ such that the
closed-loop system \revision{is well-posed,} locally exponentially
stable, and its trajectories stay in $\mathcal{Z}$. The next result
lays out LMIs for control synthesis.
\begin{theorem}[Controller Synthesis LMIs]\label{th:controller-LMIs}
  \revision{Let $\xi:\R \to \R$ be slope-restricted on
    $[\alpha,\beta]$, with $\alpha \beta < 0$.}
  %
  %
  Assume there exist $L_z \in \R^{m \times l}$,
  $L_u \in \R^{m \times lm}$,
  $L_{w_{\psi}} \in \R^{m \times
    l^2k_{\psi}}$,$L_{\phi} \in \R^{m \times k_{\phi}}$
  $L_{\psi} \in \R^{m \times l k_{\psi}}$, symmetric
  $0 \prec P \in \R^{l \times l}$,
  $0 \prec \tilde{\Lambda}_m \in \R^{m \times m}$,
  $0 \prec \tilde{\Lambda}_{k_{\psi}} \in \R^{k_{\psi} \times
    k_{\psi}}$, $\tilde{T}_{k_{\phi}} \in \mathcal{T}_{k_{\phi}}$,
  $\tilde{T}_{l k_{{\psi}}} \in \mathcal{T}_{lk_{\psi}}$ and a scalar
  $\nu > 0$ such that the following LMIs hold
  \begin{equation}
    \hspace{-0.2cm}
    \left[
      \begin{array}{cc|cc}
        \multicolumn{2}{c|}{\multirow{2}{*}{$Q_{\inte}$}} & X_{AP}
        &\hspace{-0.2cm} X_{B\tilde{Q}}
        \\  
        &  & X_{CP} &\hspace{-0.2cm} X_{D\tilde{Q}} \\ \hline
        \star & \star & P &\hspace{-0.2cm} 0 \\ 
        \star & \star & 0 &\hspace{-0.2cm} -{\tilde{Q}}
    \end{array}
    \right] \succ 0 , \,
    \begin{bmatrix}
        P+\nu \tilde{Q}_z & \hspace{-0.3cm} -\nu \tilde{S}_z \\
        \star & \hspace{-0.3cm} \nu \tilde{R}_z - 1
    \end{bmatrix} \preceq 0
    \label{eq:final_LMI}
    \end{equation}    
    where $Q_{\inte}$ is given by
    \begin{equation}
      Q_{\inte} = 
      \begin{bmatrix}
        P & -X_{BL} \, S_R \\
        \star & \tilde{R} - X_{DL} \, S_R - (X_{DL} \, S_R)^T
      \end{bmatrix}
    \end{equation}
    and ${\tilde{S}} = S_L \, S_R$ is factorized with
    $\widehat{S}_z = \tilde{Q}_z^{-1} \, \tilde{S}_z$ as
\begin{align*}
  S_L &\hspace{-0.08cm}=\hspace{-0.08cm} \diag[\, \tilde{Q}_z \otimes \tilde{\Lambda}_m,
        \diag_l \, (\tilde{Q}_z \otimes \tilde{\Lambda}_{k_{\psi}}),
  \tilde{T}_{k_{\phi}}, \, \tilde{T}_{l k_{\psi}}                 
  ] \\
  S_R &\hspace{-0.08cm}=\hspace{-0.08cm} \diag[
    \widehat{S}_z \hspace{-0.08cm}\otimes\hspace{-0.08cm} I_m ,
    \diag_l (\widehat{S}_z \hspace{-0.08cm}\otimes\hspace{-0.08cm} I_{k_{\psi}}),
    \frac{\alpha \hspace{-0.08cm}+\hspace{-0.08cm} \beta}{(\alpha \hspace{-0.08cm}-\hspace{-0.08cm} \beta)^2}I_{k_{\phi}}, \frac{\alpha \hspace{-0.08cm}+\hspace{-0.08cm} \beta}{(\alpha \hspace{-0.08cm}-\hspace{-0.08cm} \beta)^2}I_{l k_{\psi}}
  ]
\end{align*}
and where
\begin{small}  
\begin{align*}
  &X_{DL} \hspace{-0.08cm}=\hspace{-0.08cm} 
  \left[  \begin{matrix}
    L_u &\hspace{-0.08cm} L_{w_{\psi}} &\hspace{-0.08cm} L_{\phi} &\hspace{-0.08cm} L_{\psi} \\
      0 &\hspace{-0.08cm} 0 &\hspace{-0.08cm} 0 &\hspace{-0.08cm} \tilde{T}_{l k_{\psi}} \\
      G^{\phi}L_u &\hspace{-0.08cm} G^{\phi} L_{w_{\psi}} &\hspace{-0.08cm} G^{\phi}L_{\phi} + F^{\phi} \tilde{T}_{k_{\phi}} &\hspace{-0.08cm} G^{\phi} L_{\psi} \\
      G^{\psi}L_u &\hspace{-0.08cm} G^{\psi} L_{w_{\psi}} &\hspace{-0.08cm} G^{\psi}L_{\phi} &\hspace{-0.08cm} G^{\psi} L_{\psi} + F^{\psi} \tilde{T}_{lk_{\psi}}      
      \end{matrix} \right] \\
    &X_{D\tilde{Q}} = X_{DL} \, \diag \left[ I \,\,\, I \,\,\, \frac{2\alpha \beta}{(\alpha - \beta)^2}I \,\,\, \frac{2\alpha \beta}{(\alpha - \beta)^2}I  \right] \\
  &X_{AP} = \mathcal{A} P + \mathcal{B} L_z \,, \quad 
  X_{CP} = \text{vec} \left[  
  L_z \,|\, 0  \,|\, G^{\phi} L_z \,|\, G^{\psi} L_z   \right] \\
  &X_{BL} \hspace{-0.08cm}=\hspace{-0.08cm} \left[  
    \tilde{D}+J^{\psi} \,|\,H^{\psi} \,|\, H^{\phi} \,|\, H^{\psi}_*
     \right]
    S_L \hspace{-0.08cm}+\hspace{-0.08cm} 
  \left[  
    \mathcal{B} L_u \,|\, \mathcal{B} L_{w_{\psi}} \,|\, \mathcal{B} L_{\phi} \,|\, \mathcal{B} L_{\psi}
   \right]  \\
  &X_{B\tilde{Q}} =  X_{BL} \, \diag \left[ I \,\,\, I \,\,\, \frac{2\alpha \beta}{(\alpha - \beta)^2}I \,\,\, \frac{2\alpha \beta}{(\alpha - \beta)^2}I  \right].
\end{align*}
\end{small}
Then, the gains
\begin{alignat*}{2}
  K_u &= L_u \, (\tilde{Q}_z^{-1}  \otimes \tilde{\Lambda}_m^{-1} ), \,\,
  K_{w_{\psi}} = L_{w_{\psi}} \, \diag_l \, (\tilde{Q}_z ^{-1}  \otimes
        \tilde{\Lambda}_{k_{\psi}} ^{-1} ) , \\
  K_z &= L_z \, P^{-1} , \,\,
  K_{\phi} = L_{\phi} \, \tilde{T}^{-1}_{k_{\phi}} , \,\,
  K_{\psi} = L_{\psi} \, \tilde{T}^{-1}_{l k_{\psi}}
\end{alignat*}
ensure that~\eqref{eq:LFR_closed_loop} \revision{is well-posed and}
locally exponentially stable, with region of attraction
$\mathcal{Z}_{ROA} = \{ z \in \R^l \, \big| \, z^T P^{-1}z \leq 1 \}$
\revision{being forward invariant and
  $\mathcal{Z}_{ROA} \subseteq \mathcal{Z}$.}
\end{theorem}
%
%
%
\begin{proof}
  We begin by showing that $V(z) = z^T P^{-1} z$ is a Lyapunov
  function for the closed-loop system \eqref{eq:LFR_closed_loop}.
  %
  %
  This is done by applying \cite[Thm 2]{CWS:01} and the dualization
  lemma \cite[Lemma A.1]{CWS:01} since $\alpha \beta < 0$ to obtain the quadratic inequality
  {\small
    \begin{equation*}
      M = 
      \left[\def\arraystretch{1.0}
        \begin{array}{cc}
          A^T & C^T \\
          -I & 0    \\ 
          B^T & D^T \\
          0 & -I   
        \end{array}
      \right] ^T
      \hspace{-0.2cm}
      \left[\def\arraystretch{1.0}
        \begin{array}{cccc}
          -P & 0 & 0 & 0 \\
            0 & P & 0 & 0 \\ 
          0 & 0 & {\tilde{Q}} & {\tilde{S}} \\
          0 & 0 & {\tilde{S}}^T& {\tilde{R}}    
        \end{array}
      \right]
      \hspace{-0.2cm}
      \left[\def\arraystretch{1.0}
        \begin{array}{cc}
          A^T & C^T \\
          -I & 0    \\ 
          B^T & D^T \\
          0 & -I   
        \end{array}
      \right] \succ 0     
    \end{equation*}   
  }
  %
  %
  Expanding the expression we obtain two terms
  \begin{equation*}
    M = 
    \begin{bmatrix}
        \star
    \end{bmatrix}^T 
    \begin{bmatrix}
        -P &\hspace{-0.2cm} 0 \\
        0 &\hspace{-0.2cm} P 
    \end{bmatrix} 
    \begin{bmatrix}
        A^T &\hspace{-0.2cm} C^T \\
        -I &\hspace{-0.2cm} 0 
    \end{bmatrix} \hspace{-0.1cm}+\hspace{-0.1cm} 
    \begin{bmatrix}
        \star
    \end{bmatrix}^T \hspace{-0.1cm}
    \begin{bmatrix}
        {\tilde{Q}} &\hspace{-0.2cm} {\tilde{S}} \\
        {\tilde{S}}^T &\hspace{-0.2cm} {\tilde{R}}
    \end{bmatrix} \hspace{-0.1cm}
    \begin{bmatrix}
        B^T &\hspace{-0.2cm} D^T \\
        0 &\hspace{-0.2cm} -I 
    \end{bmatrix}.    
\end{equation*}
Now, in order to take care of the quadratic terms, we have 
{\small
\begin{align*}
    &M = 
    \begin{bmatrix}
        \star
    \end{bmatrix}^T 
    \begin{bmatrix}
        0 & 0 \\
        0 & P 
    \end{bmatrix} 
    \begin{bmatrix}
        0 & 0 \\
        -I & 0 
    \end{bmatrix} 
    - 
    \begin{bmatrix}
        \star
    \end{bmatrix} 
    P^{-1} 
    \begin{bmatrix}
        X_{AP} \\
        X_{CP}
    \end{bmatrix}^T + \\
    &
    \begin{bmatrix}
        \star
    \end{bmatrix}^T 
    \begin{bmatrix}
        0 & {\tilde{S}} \\
        \star & {\tilde{R}}
    \end{bmatrix} 
    \begin{bmatrix}
        B^T & D^T \\
        0 & -I 
    \end{bmatrix} 
    + 
    \begin{bmatrix}
        \star
    \end{bmatrix} 
    \tilde{Q}^{-1} 
    \begin{bmatrix}
        X_{B\tilde{Q}} \\
        X_{D\tilde{Q}}
    \end{bmatrix}^T.
\end{align*}
}
We use the factorization $\tilde{S} = S_L \, S_R$ to write
{\small
\begin{equation*}
\begin{bmatrix}
    \star
\end{bmatrix}^T \, 
\begin{bmatrix}
    0 & {\tilde{S}} \\
    \star & {\tilde{R}}
\end{bmatrix} \, 
\begin{bmatrix}
    B^T & D^T \\
    0 & -I 
\end{bmatrix} =     
\begin{bmatrix}
    \star
\end{bmatrix}^T
\begin{bmatrix}
    0 & S_R \\
    \star & {\tilde{R}}
\end{bmatrix}
\begin{bmatrix}
    X_{BL}^T & X_{DL}^T \\
    0 & -I
\end{bmatrix} 
\end{equation*}
}
Then, the matrix $M$ can be written as
{\small
\begin{align*}
&M = 
    \begin{bmatrix}
        \star
    \end{bmatrix}^T
    \begin{bmatrix}
        0 & 0 \\
        0 & P
    \end{bmatrix}
    \begin{bmatrix}
        0 & 0 \\
        -I & 0
    \end{bmatrix} +
\begin{bmatrix}
    \star
\end{bmatrix}\, 
\tilde{Q}^{-1} \, 
\begin{bmatrix}
    X_{B\tilde{Q}} \\
    X_{D\tilde{Q}}
\end{bmatrix}^T +
  \\
& \quad \begin{bmatrix}
    \star
\end{bmatrix}^T
\begin{bmatrix}
    0 & S_R \\
    \star & {\tilde{R}}
\end{bmatrix}
\begin{bmatrix}
    X_{BL}^T & X_{DL}^T \\
    0 & -I
\end{bmatrix}
 - \,
\begin{bmatrix}
    \star
\end{bmatrix}\, 
P^{-1} \, 
\begin{bmatrix}
    X_{AP} \\
    X_{CP}
\end{bmatrix}^T \\
&= Q_{\inte} + \begin{bmatrix}
    \star
\end{bmatrix}\, 
\tilde{Q}^{-1} \, 
\begin{bmatrix}
    X_{B\tilde{Q}} \\
    X_{D\tilde{Q}}
\end{bmatrix}^T \, - \,
\begin{bmatrix}
    \star
\end{bmatrix}\, 
P^{-1} \, 
\begin{bmatrix}
    X_{AP} \\
    X_{CP}
\end{bmatrix}^T 
\end{align*}
} which upon taking the Schur complement~\cite{CKL-RM:04}
%
%
twice gives the left inequality in equation \eqref{eq:final_LMI}. The
feasibility of this LMI then implies that $V(z) = z^T P^{-1}z$ is a
Lyapunov function.

After applying the dualization lemma \cite[Lemma A.1]{CWS:01} to the
right inequality in \eqref{eq:final_LMI}, we obtain
\begin{equation}
  \begin{bmatrix}
    Q_z & S_z \\
    S_z^T & R_z 
  \end{bmatrix} - \nu 
  \begin{bmatrix}
    -P^{-1} & 0 \\
    0 & 0 
  \end{bmatrix} \succ 0 .
\end{equation}
Applying the S-procedure~\cite{CS-SW:00} shows that
$\mathcal{Z}_{ROA}
\subseteq \mathcal{Z}$. Since $\mathcal{Z}_{ROA}$ is forward
invariant, closed-loop trajectories that start there remain in
$\mathcal{Z}$ for all time. \revision{Finally, well posedness of the
  closed-loop system and hence of the implicit equations
  \eqref{eq:controller_implicit} is ensured if the LMI
  \eqref{eq:final_LMI} is feasible, cf.~\cite[Thm 2]{CWS:01}.}
\end{proof}
\revision{
  \begin{remark}[LMI for $\alpha = 0$]
    %
    %
    For the special case $\alpha = 0$ (e.g., ReLU activation), the
    left matrix in \eqref{eq:final_LMI} becomes singular. Instead, one
    can treat the case $\alpha = 0$ separately, adapting the proof of
    Theorem~\ref{th:controller-LMIs} to obtain a new LMI,
    corresponding to removing the last $k_{\phi} + lk_{\psi}$ rows and
    columns from the left matrix in \eqref{eq:final_LMI}, that ensures
    stability and avoids numerical issues associated with the
    singularity. \oprocend
\end{remark}
}

%
%
\revision{ We conclude this section illustrating the controller design
  in the 4-dimensional generalized bilinear system
\begin{equation}
  \label{eq:example_system}
  x^+ = A + B_0 u +
  \sum_{i=1}^{2}\bigl(u_i B_i + \phi_i(u_i)\,C_i\bigr)\,x
\end{equation}
where $x \in \R^4, \, u \in \R^2$, $\phi_1(u_1) = e^{u_1} - 1$ and
$\phi_2(u_2) = \frac{u_2^3}{1 + u_2^2}$ are approximated by ReLU MLP
neural networks with hidden sizes = [10,10].  The matrices are given
by
\begin{small}
\begin{equation*}
  \setlength{\arraycolsep}{0.6pt}
  \def\arraystretch{1}
  \begin{aligned}
    &A = 
      \left[
  \begin{array}{@{}cccc@{}}
    1    & 0.3  & 0.4   & 0.1  \\
    1    & -0.2 & 0     & 0.05 \\
    0    & 1.2  & -0.5  & 0.02 \\
    0    & 0    & 0     & 0.2
  \end{array}
\right], 
B_1 = 
\left[
  \begin{array}{@{}cccc@{}}
    0 & 0 & 0 & 0\\
    0 & 0 & 0 & 0\\
    0 & 0 & 0 & 0\\
    0 & 0 & 0 & -0.5
  \end{array}
\right], 
B_2 = 
\left[
  \begin{array}{@{}cccc@{}}
    -0.3 & 0    & 0 & 0\\
     0.3 & 0    & 0 & 0\\
     0   & 0    & 0 & 0\\
     0   & 0    & 0 & 0
  \end{array}
\right],
\\
&B_0 = 
\left[
  \begin{array}{@{}cc@{}}
    1    & 0 \\  
    0    & 1 \\  
    0    & 1 \\  
   -1    & 0  
  \end{array}
\right],
C_1 = -
\left[
  \begin{array}{@{}cccc@{}}
    0    & 0    & 0    & 0    \\
    0    & 0    & 0    & 0    \\
    0    & 0    & 0    & 0    \\
   0.6  & 0.3 & 0.3 & 1.2
  \end{array}
  \right],
  C_2 = 
  \left[
  \begin{array}{@{}cccc@{}}
    -0.3  & 0    & 0     & 0    \\
    0.3  & 0    & 0     & 0    \\
    0    & 0    & 0     & 0    \\
    0    & 0    & -0.15 & 0
  \end{array}
  \right] .
  \end{aligned} 
\end{equation*}
\end{small}
} We conduct simulations in Python using the \emph{CVXPY} toolbox
\cite{SD-SB:16} with the SDP solver \emph{MOSEK}~\cite{aps2022mosek}.
We maximize $\trace(P)$ such that the LMIs are feasible.
\revision{The target equilibrium is $z_* = 0,u_* = 0$ and
  $\mathcal{Z} = \{z \in \R^4 \mid \|z\|_2 \leq
  0.08\}$. Figure~\ref{fig:simulation_results} compares the results
  obtained under our approach with the baseline gain-scheduled
  controller in \cite{RS-JB-FA:23}, which works with the bilinear
  system and ignores the neural network nonlinearities
  $\phi_1$,~$\phi_2$. Beyond the improvement in performance, our
  approach is able to provide a guaranteed region of attraction
  $\mathcal{Z}_{ROA}$, with $\trace(P) = 0.1959$.}
\vspace*{-2ex}  
\begin{figure}[htb]
  \centering
  \includegraphics[width=\linewidth]{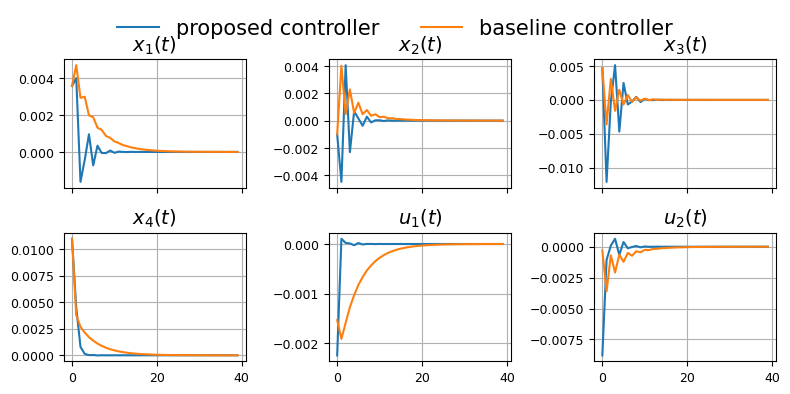}
  \vspace*{-4ex}
  \caption{Time evolution of~\eqref{eq:example_system}
    using the controller \eqref{eq:controller_implicit}, which
    accounts for the nonlinearities $\phi_1$ and $\phi_2$. Results are
    compared with the controller from \cite{RS-JB-FA:23}, which does
    not take into account these
    nonlinearities.}\label{fig:simulation_results}
  \vspace*{-1ex}
\end{figure}
%
%
%
%
%
%

\begin{remark}\longthmtitle{Computational Complexity and LMI
    Feasibility} \revision{LMI-based controller synthesis can
    become computationally expensive for high-dimensional
    systems. This can be mitigated by neural network
    pruning~\cite{DB-JJJO-JF-JG:20} and by leveraging
    structure-exploiting solvers~\cite{DG-TH-CWS-CE:23} to reduce
    problem size and computation time. Moreover, feasibility of LMIs
    is a well-known open problem. Integral quadratic constraints
    (IQCs) can help reduce conservatism and improve
    feasibility.~Recently,~\cite{HG-TY-YE-VM-DP-ST:24} shows that, for
    LTI-NFLs, LMI infeasibility implies a destabilizing nonlinearity
    within the considered class exists.  Extending such IQC-based
    feasibility results to bilinear-NFLs is a promising direction for
    future research.  \oprocend
  }
\end{remark}

\section{Conclusions}
We have presented an approach to locally exponentially stabilize a
bilinear-NFL system. Our methodology is based on exactly reformulating
the system with a linear fractional representation and abstracting out
all nonlinearities using fabricated uncertainties, for which we find
appropriate over-approximating sets using quadratic constraints. The
resulting system is then amenable to control synthesis via LPV design.
Future work will reduce conservatism by using ICQs, incorporating rate
bounds on parameters and neural network pruning for defining the
uncertainty sets.
%

%
%

\bibliographystyle{ieeetr}
\bibliography{../bib/alias,../bib/Main-add,../bib/Main,../bib/JC}

\begin{thebibliography}{10}

\bibitem{ME-GH-CS-JPH:21}
M.~Everett, G.~Habibi, C.~Sun, and J.~P. How, ``Reachability analysis of neural feedback loops,'' {\em IEEE Access}, vol.~9, pp.~163938--163953, 2021.

\bibitem{HH-MF-MM-GJP:20}
H.~Hu, M.~Fazlyab, M.~Morari, and G.~J. Pappas, ``Reach-{SDP}: Reachability analysis of closed-loop systems with neural network controllers via semidefinite programming,'' in {\em {IEEE} Conf.\ on Decision and Control}, (Jeju Island, Republic of Korea), pp.~5929--5934, Dec. 2020.

\bibitem{CH-JF-WL-XC-QZ:19}
C.~Huang, J.~Fan, W.~Li, X.~Chen, and Q.~Zhu, ``Reach{NN}: Reachability analysis of neural-network controlled systems,'' {\em ACM Transactions on Embedded Computing Systems}, vol.~18, no.~5s, pp.~1--22, 2019.

\bibitem{SD-XC-SS:19}
S.~Dutta, X.~Chen, and S.~Sankaranarayanan, ``Reachability analysis for neural feedback systems using regressive polynomial rule inference,'' in {\em Hybrid Systems: Computation and Control}, pp.~157--168, 2019.

\bibitem{SD-NM-BR-VY:20}
S.~Dean, N.~Matni, B.~Recht, and V.~Ye, ``Robust guarantees for perception-based control,'' in {\em Learning for Dynamics and Control}, pp.~350--360, PMLR, 2020.

\bibitem{NJ-MA-PS:24margins}
N.~Junnarkar, M.~Arcak, and P.~Seiler, ``Stability margins of neural network controllers,'' {\em arXiv preprint arXiv:2409.09184}, 2024.

\bibitem{NJ-MA-PS:24synthesis}
N.~Junnarkar, M.~Arcak, and P.~Seiler, ``Synthesizing neural network controllers with closed-loop dissipativity guarantees,'' {\em arXiv preprint arXiv:2404.07373}, 2024.

\bibitem{GP-AA-DG-LL-AHR-THS:25}
G.~Pillonetto, A.~Aravkin, D.~Gedon, L.~Ljung, A.~H. Ribeiro, and T.~B. Sch{\"o}n, ``Deep networks for system identification: a survey,'' {\em Automatica}, vol.~171, p.~111907, 2025.

\bibitem{MH-JC:24-auto}
M.~Haseli and J.~Cort\'es, ``Modeling nonlinear control systems via {Koopman} control family: universal forms and subspace invariance proximity,'' {\em Available at \url{https://arxiv.org/abs/2307.15368}}, 2023.

\bibitem{RS-JB-FA:23}
R.~Str{\"a}sser, J.~Berberich, and F.~Allg\"ower, ``Control of bilinear systems using gain-scheduling: {Stability} and performance guarantees,'' in {\em {IEEE} Conf.\ on Decision and Control}, (Singapore), pp.~4674--4681, Dec. 2023.

\bibitem{LEG-FG-BT-AA-AT:21}
L.~El~Ghaoui, F.~Gu, B.~Travacca, A.~Askari, and A.~Tsai, ``Implicit deep learning,'' {\em SIAM Journal on Mathematics of Data Science}, vol.~3, no.~3, pp.~930--958, 2021.

\bibitem{MF-AR-HH-MM-GP:19}
M.~Fazlyab, A.~Robey, H.~Hassani, M.~Morari, and G.~Pappas, ``Efficient and accurate estimation of {L}ipschitz constants for deep neural networks,'' {\em Advances in Neural Information Processing Systems}, vol.~32, 2019.

\bibitem{KZ-JCD:98}
K.~Zhou and J.~C. Doyle, {\em Essentials of Robust Control}, vol.~104.
\newblock Prentice hall Upper Saddle River, NJ, 1998.

\bibitem{CS-SW:00}
C.~Scherer and S.~Weiland, ``Linear matrix inequalities in control,'' {\em Lecture Notes, Dutch Institute for Systems and Control, Delft, The Netherlands}, vol.~3, no.~2, 2000.

\bibitem{CWS:01}
C.~W. Scherer, ``{LPV} control and full block multipliers,'' {\em Automatica}, vol.~37, no.~3, pp.~361--375, 2001.

\bibitem{CKL-RM:04}
C.~K. Li and R.~Mathias, ``Extremal characterizations of the {Schur} complement and resulting inequalities,'' {\em SIAM Review}, vol.~42, no.~2, pp.~233--246, 2000.

\bibitem{SD-SB:16}
S.~Diamond and S.~Boyd, ``{CVXPY}: A python-embedded modeling language for convex optimization,'' {\em Journal of Machine Learning Research}, vol.~17, no.~83, pp.~1--5, 2016.

\bibitem{aps2022mosek}
M.~ApS, ``Mosek optimizer {API} for python,'' {\em Version}, vol.~9, no.~17, pp.~6--4, 2022.

\bibitem{DB-JJJO-JF-JG:20}
D.~Blalock, J.~J. Gonzalez~Ortiz, J.~Frankle, and J.~Guttag, ``What is the state of neural network pruning?,'' {\em Proceedings of Machine Learning and Systems}, vol.~2, pp.~129--146, 2020.

\bibitem{DG-TH-CWS-CE:23}
D.~Gramlich, T.~Holicki, C.~W. Scherer, and C.~Ebenbauer, ``A structure exploiting {SDP} solver for robust controller synthesis,'' {\em IEEE Control Systems Letters}, vol.~7, pp.~1831--1836, 2023.

\bibitem{HG-TY-YE-VM-DP-ST:24}
H.~Gyotoku, T.~Yuno, Y.~Ebihara, V.~Magron, D.~Peaucelle, and S.~Tarbouriech, ``On dual of {LMI}s for absolute stability analysis of nonlinear feedback systems with static {O'Shea-Zames-Falb Multipliers},'' {\em arXiv preprint arXiv:2411.14339}, 2024.

\end{thebibliography}

\end{document}